\documentclass[conference,a4paper,10pt]{IEEEtran}
\usepackage{bbm,dsfont,mathrsfs,fixmath,amsthm}
\usepackage{amsmath,amssymb,graphicx} 
\usepackage{stmaryrd,cite}
\usepackage{amsmath,amstext,amsfonts,amssymb}
\usepackage{epsfig}
\usepackage{exscale}
\usepackage{enumerate}
\usepackage{mathtools}
\usepackage{multirow}
\usepackage{booktabs}

\usepackage{algorithm,algorithmic}

\usepackage{amssymb}
\usepackage{amscd}
\usepackage{amsmath}
\usepackage{bm}
\usepackage{epsfig}
\usepackage{graphics}
\usepackage{psfrag}
\usepackage{rotating}
\usepackage{amsmath}
\usepackage{amsfonts}
\usepackage{url}
\usepackage{color}
\usepackage{epstopdf}
\usepackage{amsthm}
\usepackage{tikz}
\usepackage{mathtools}
\usetikzlibrary{arrows,shapes,decorations,backgrounds}
\usepackage[final]{pdfpages}
\usepackage{multicol}
\usepackage{cite}
\usepackage{graphicx}
\usepackage{caption}
\usepackage{subfig}
\usepackage{breqn}

\newtheorem{theorem}{Theorem}
\newtheorem{lemma}[theorem]{Lemma}

\newtheorem{remark}{Remark}
\newtheorem{definition}{Definition}

\newcommand{\T}{{\scriptscriptstyle\mathsf{T}}}
\renewcommand{\H}{{\scriptscriptstyle\mathsf{H}}}

\newcommand{\mmse}{\mathsf{mmse}}

\newfont{\bbb}{msbm10 scaled 500}

\newfont{\bb}{msbm10 scaled 1100}
\newcommand{\CC}{\mbox{\bb C}}
\newcommand{\RR}{\mbox{\bb R}}

\newcommand{\EE}{\mbox{\bb E}}

\newcommand{\zerov}{{\bf 0}}

\newcommand{\Am}{{\bf A}}
\newcommand{\Bm}{{\bf B}}

\newcommand{\Id}{{\bf I}}
\newcommand{\Jm}{{\bf J}}

\newcommand{\Pm}{{\bf P}}
\newcommand{\Qm}{{\bf Q}}

\newcommand{\Um}{{\bf U}}
\newcommand{\Wm}{{\bf W}}
\newcommand{\Vm}{{\bf V}}
\newcommand{\Xm}{{\bf X}}
\newcommand{\Ym}{{\bf Y}}
\newcommand{\Zm}{{\bf Z}}

\newcommand{\Cc}{{\cal C}}

\newcommand{\Lc}{{\cal L}}

\newcommand{\Nc}{{\cal N}}

\newcommand{\Uc}{{\cal U}}

\newcommand{\Xc}{{\cal X}}
\newcommand{\Yc}{{\cal Y}}

\newcommand{\betav}{\hbox{\boldmath$\beta$}}

\newcommand{\etav}{\hbox{\boldmath$\eta$}}

\newcommand{\thetav}{\hbox{\boldmath$\theta$}}

\newcommand{\rhov}{\hbox{\boldmath$\rho$}}

\newcommand{\Sigmam}{\hbox{\boldmath$\Sigma$}}

\newcommand{\Omegam}{\hbox{\boldmath$\Omega$}}

\newcommand{\eqdef}{\stackrel{\Delta}{=}}

\newcommand{\cov}{{\hbox{cov}}}

\renewcommand{\Am}{\pmb{A}}
\renewcommand{\Bm}{\pmb{B}}

\renewcommand{\Jm}{\pmb{J}}

\renewcommand{\Pm}{\pmb{P}}
\renewcommand{\Qm}{\pmb{Q}}

\renewcommand{\Um}{\pmb{U}}
\renewcommand{\Vm}{\pmb{V}}
\renewcommand{\Wm}{\pmb{W}}
\renewcommand{\Xm}{\pmb{X}}
\renewcommand{\Ym}{\pmb{Y}}
\renewcommand{\Zm}{\pmb{Z}}

\newcommand{\mkv}{-\!\!\!\!\minuso\!\!\!\!-}

\include{macros}

\begin{document}

\title{Scalable Vector Gaussian Information Bottleneck}

\author{\IEEEauthorblockN{Mohammad Mahdi Mahvari}
\IEEEauthorblockA{
\thanks{}
Sharif University of Technology \\
Tehran, Iran\\
 {\tt mahdimahvar@gmail.com}
}
\and
\IEEEauthorblockN{Mari Kobayashi}
\IEEEauthorblockA{
\thanks{The work of M. Kobayashi was supported by DFG.}
Technical University of Munich \\
Munich, Germany\\
 {\tt mari.kobayashi@tum.de}
}
\and
\IEEEauthorblockN{Abdellatif Zaidi}
\IEEEauthorblockA{
\thanks{}
Universite Paris-Est \\
Champs-sur-Marne, 77454, France\\
 {\tt abdellatif.zaidi@u-pem.fr}
}
}

\maketitle
\begin{abstract}
In the context of statistical learning, the Information Bottleneck method seeks a right balance between accuracy and generalization capability through a suitable tradeoff between compression complexity, measured by minimum description length, and distortion evaluated under logarithmic loss measure.  In this paper, we study a variation of the problem, called scalable information bottleneck, in which the encoder outputs multiple descriptions of the observation with increasingly richer features. The model, which is of successive-refinement type with degraded side information streams at the decoders, is motivated by some application scenarios that require varying levels of accuracy depending on the allowed (or targeted) level of complexity. We establish an analytic characterization of the optimal relevance-complexity region for vector Gaussian sources. Then, we derive a variational inference type algorithm for general sources with unknown distribution; and show means of parametrizing it using neural networks. Finally, we provide experimental results on the MNIST dataset which illustrate that the proposed method generalizes better to unseen data during the training phase.
\end{abstract}

\section{Introduction}
In statistical (supervised) learning models, one seeks to strike a right balance between accuracy and generalization capability. Many existing machine learning algorithms generally fail to do so; or else only at the expense of large amounts of algorithmic components and hyperparameters that are heavily tuned heuristically for a given task (and, even so, generally fall short of generalizing to other tasks). 
The Information Bottleneck (IB) method~\cite{tishby1999information}, which is essentially a remote source coding problem under logarithmic loss fidelity measure, 
attempts to do so by finding a suitable tradeoff between algorithm's robustness, measured by compression complexity, and 
accuracy or relevance, measured by the allowed average logarithmic loss (see e.g. ~\cite{zaidi2020information, goldfeld2020information}). More specifically, for a target (label) variable $X$ and an observed variable $Y$, the IB finds a description $U$ that is maximally informative about $X$ while being minimally informative about $Y$, where informativeness is measured via Shannon's mutual information. The relevance-complexity tradeoff has been studied in various setups, including multivariate IB \cite{friedman2001multivariate}, distributed IB \cite{ugur2020vector}, IB with decoder side information
\cite{ugur2020vector,tian2009remote}, IB with unknown fading channels \cite{steiner2020broadcast}. Moreover, the IB model has been applied to communication scenarios with oblivious relays \cite{winkelbauer2014rate,caire2018information,aguerri2019capacity,xu2021information}.
Recently, the IB framework has been extended to a practical scenario when the source and observation distribution is unknown and can be estimated empirically \cite{alemi2016deep, kingma2013auto,aguerri2019distributed,zaidi2020information, goldfeld2020information}. 

\begin{figure}[t]
\vspace*{-0.8em}	
\begin{center}	
\includegraphics[width=0.38\textwidth]{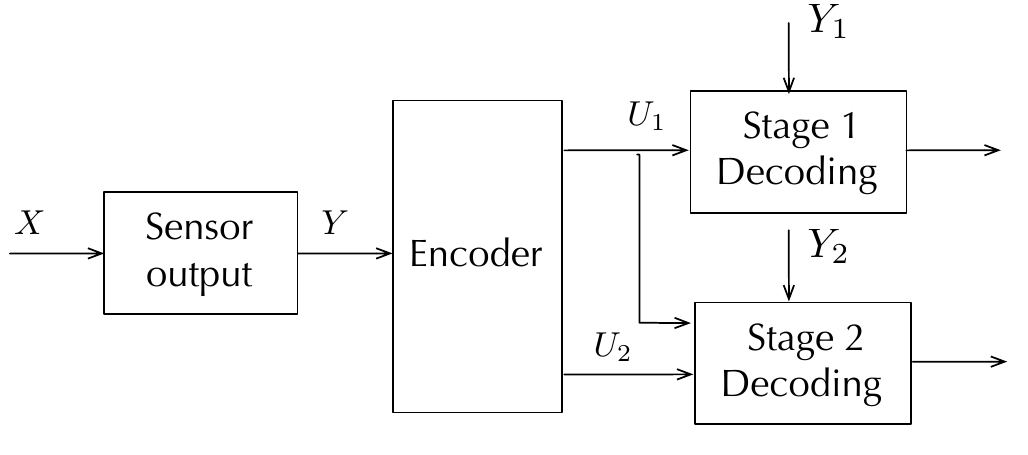}
\caption{Scalable information bottleneck for $T=2$ stages.}
\label{fig:ScalableRSC}
\end{center}
\vspace*{-1.5em}	
\end{figure} 
We study a variation of the IB problem, called {\it scalable information bottleneck}, where the encoder outputs $T\geq 2$ descriptions with increasingly richer features. 
This is motivated primarily by application scenarios in which a varying level of accuracy is required depending on the allowed and/or required level of complexity. The model is illustrated in Fig.~\ref{fig:ScalableRSC} for $T=2$. As an example, one may think about the simple scenario of making inference about a moving object on the road. A decision maker which receives only coarse information about the object would have to merely identify its type (e.g., car, bicycle, bus), while one that receives also refinement information would have to infer more accurate description.


In the $T$-stage scalable information bottleneck, the encoder observes the source $X$ via a sensor output $Y$ and wishes to encode $Y$ into $T$ stages 
of descriptions denoted by $(U_1, \dots, U_T)$,  while the $t$-stage decoder wishes to reconstruct the source from $(U_1, \dots, U_t)$ and its side information $Y_t$. For the case of vector Gaussian sources and channels with degraded side information, we fully characterize the relevance-complexity region. Numerical examples for simple scalar Gaussian sources and channels illustrate the usefulness of decoder side information. For a more practical case where the joint source and channel distribution is unknown, we derive a variational inference type algorithm with a set of training data. The experiments using MNIST dataset demonstrate that our proposed scheme can be efficiently applied to the pattern classification by offering stronger generalization capability than a single-stage case. 

Remark that a similar successive refinement model with degraded side information has been studied in the context of the source coding problem \cite{tian2007multistage,xu2019vector}. Based on the rate-distortion region under general distortion measure \cite[Theorem 1]{tian2007multistage}, a recent work characterized the rate-distortion region of the vector Gaussian sources and channels under the quadratic distortion \cite{xu2019vector}. Similarly,  we adapt \cite[Theorem 1]{tian2007multistage} to the remote source setup and the logarithmic loss distortion measure, relevant to the classification problem. 


\section{Problem Formulation}\label{Problem Formulation}

Let $(X^n, Y^n, Y_1^n, \dots, Y_T^n) \in \Xc^n \times \Yc^n \times \Yc_1^n \times \dots \times \Yc_T^n$ be a sequence of $n$ i.i.d. discrete random variables corresponding to the source, the observation, and side information at 
$T$ stages. In particular, we consider degraded side information satisfying the following Markov chain
\begin{align}
(X, Y) \mkv Y_T \mkv \dots  \mkv Y_1.
\end{align}
\begin{definition}
A $T$-stage successive refinement code of length $n$ consists of $T$ encoding functions
\begin{align}
\phi_t^{(n)}: \Yc^n \mapsto \{1, \dots, M_t^{(n)}\},
\end{align}
and $T$ decoding functions
\begin{align}
\psi_t^{(n)}: \{1, \dots, M_1^{(n)}\} \times \dots \{1, \dots, M_t^{(n)}\} \times \Yc_t \mapsto \hat{\Xc}^n_t,
\end{align}
for any $t\in[T]$, where the reconstruction $\hat{\Xc}^n_t$ is the set of probability distributions over the $n$-Cartesian product of $\RR$. 
\end{definition}
\begin{definition}
A $T$-stage successive refinement code $(\Delta_1, \dots, \Delta_T, R_1, \dots, R_T)$ for the $T$-stage IB problem is achievable if there exists $n$, $T$ encoding functions,
and $T$ decoding functions such that
\begin{align}
\Delta_t & \leq \frac{1}{n} I\left(X^n; \psi_t^{(n)}(\phi_1^{(n)}(Y^n), \dots, \phi_t^{(n)}(Y^n), Y_t^n) \right) \nonumber,\\
R_t & \geq \frac{1}{n} \log M_t^{(t)},\nonumber
\end{align}
for any $t\in[T]$. The relevance-complexity region of the $T$-scalable IB problem is defined as the union of all non-negative tuples $(\Delta_1, \dots, \Delta_T,R_1, \dots, R_T)$ that are 
achievable. 
\end{definition}
Next, we provide the relevance-complexity region of the $T$-scalable information bottleneck problem. 
By a straightforward adaptation of \cite[Theorem 1]{tian2007multistage} to the remote source setup and the logarithmic loss distortion measure considered in this work, we obtain the following result. Namely, 
the relevance-complexity region 
for the $T$-stage successive refinement code satisfies
\begin{subequations}\label{eq:region1}
\begin{align}
\sum_{l=1}^{t}R_l  &\geq  I(Y; U_1, \dots , U_t |Y_t),\;\;\; \forall t=1, \dots, T \label{eq:region1-rate}\\
\Delta_t  &  \leq I(X;U_1, \dots, U_t, Y_t), \;\;\; \forall t=1, \dots, T\label{eq:region1-relevance}
\end{align}
\end{subequations}
for some joint pmf
\begin{align}
&P_{XY} (x,y)P_{Y_T|XY}(y_T|x, y) \nonumber\\
& \cdot \sum_{t=1}^{T-1} P_{Y_t|Y_{t+1}}(y_t| y_{t+1})P_{U_1, \dots, U_T |Y}(u_1,\dots,u_T|y).
\label{eq:joint pdf}
\end{align}
\begin{remark}\label{Changing Relevance Formulation}
Similarly to \cite[Remark 1]{tian2007multistage}, we can provide an alternative region to \eqref{eq:region1} 
under the stronger Markov chain $(X, Y_1, \dots, Y_T) \mkv Y \mkv U_T \mkv \dots \mkv U_1$.
\begin{subequations}\label{eq:region2}
\begin{align}
\sum_{l=1}^{t}R_l  &\geq  I(Y; U_1, \dots , U_t |Y_t), \;\;\; \forall t \label{eq:region2-rate} \\
\Delta_t  &  \leq I(X;U_t, Y_t), \;\;\; \forall t.  \label{eq:region2-relevance} 
\end{align}
\end{subequations}
These two regions are equivalent as the RHS of \eqref{eq:region1-relevance} coincides with that of \eqref{eq:region2-relevance} under the stronger Markov chain. 
The former \eqref{eq:region1} will be used to characterize the relevance-complexity region in Section \ref{sect:Vector Gaussian with SI}, while the latter \eqref{eq:region2} will be used for a more practical variational case when the underlying distribution is unknown in Section \ref{sect:Variational}.
\end{remark}

\section{Scalable Vector Gaussian IB}\label{sect:Vector Gaussian with SI}
\subsection{System Model and Main Result}
We focus on the case when the source, the observation as well as the side informations are jointly Gaussian distributed and its distribution is known. Namely, the source is given by a $m$-dimensional Gaussian random vector denoted by $
\Xm$ of zero mean and covariance $\Sigmam_x$, while the observation $\Ym$, the side information $\Ym_t$ in stage $t$, of dimension $m\times 1$, are given respectively by 
\begin{align}\label{eq:670}
\Ym &= \Xm + \Wm_0, \;\;\;  \Ym_t = \Xm + \Wm_t \\
\Sigmam_{w_0, w_t} = \left[\begin{matrix}
\Wm_0\\ \Wm_t 
\end{matrix}\right]
& \sim\Nc_{\CC} \left( \left[
\begin{matrix} \zerov\\ \zerov \end{matrix}\right], 
\left[
 \begin{matrix}
 \Sigmam_0 & \Sigmam_{0t}\\
 \Sigmam_{t0} & \Sigmam_t
 \end{matrix}\right]
\right)
\label{eq:covariance matrix}
\end{align}
where we assume that $\Sigmam_T \preceq \dots \preceq \Sigmam_1$ holds and that $\Ym_t$, $\Ym_s$, for $t\neq s$ are independent given $\Ym$. 
We also define the covariance of the observation noise $\Wm_0$ given the side information noise $\Wm_t$ as
\begin{align}\label{eq:vect_crosscovarianc}
\Sigmam_{0|t} \eqdef \cov(\Wm_0|\Wm_t) = \Sigmam_0 - \Sigmam_{0t} \Sigmam_t^{-1} \Sigmam_{t0} 
\end{align}
which reduces to $\Sigmam_0$ simply in the absence of side information or in the case of uncorrelated noises. 
\begin{theorem}
\label{Vector_Gaussian_proposition}
For the vector Gaussian model described in \eqref{eq:670}-\eqref{eq:vect_crosscovarianc}, the relevance-complexity region is given by
\begin{subequations}\label{eq:vec_region}
\begin{align}
\Delta_t  & \leq \log \left|\Id \!+\!  \left[ \Sigmam_t^{-1} \!+\! (\Id \!-\! \Sigmam_t^{-1}\Sigmam_{t0})  \Omegam_t (\Id \!-\! \Sigmam_t^{-1}\Sigmam_{t0})  \right] \! \Sigmam_x \right| \label{eq:vec_region_a}\\
\Delta_t  & \leq \sum_{l=1}^{t}R_l + \log |\Id- \Omegam_t \Sigmam_{0|t} | + \log|\Id + \Sigmam_x \Sigmam_t^{-1} |,\label{eq:vec_region_b}
\end{align}
\end{subequations}
for any $t\in [T]$, where $\zerov \preceq \Omegam_{1}  \preceq \dots \preceq \bm{\Omega}_{T} \preceq \Sigmam_{0|T}^{-1}$. 
\end{theorem}

\begin{proof}
The proof is given in Appendix \ref{proof of vector}.
\end{proof}

\begin{remark}
Our result covers some special cases known in the literature. For the case of a single stage ($T=1$), region \eqref{eq:vec_region} reduces to
\begin{subequations}\label{eq:vec_region_T1}
\begin{align}
\Delta  & \leq \log \left|\Id \!+\!  \left[ \Sigmam_1^{-1} \!+\! (\Id \!-\! \Sigmam_1^{-1}\Sigmam_{10})  \Omegam (\Id \!-\! \Sigmam_1^{-1}\Sigmam_{10})  \right] \! \Sigmam_x \right| \label{eq:vec_region_a_T1}\\
\Delta  & \leq R + \log |\Id- \Omegam \Sigmam_{0|1} | + \log|\Id + \Sigmam_x \Sigmam_1^{-1} |,\label{eq:vec_region_b_T1}
\end{align}
\end{subequations}
Furthemore, for $T=1$ and the scalar Gaussian case we can simplify the region as
\begin{align}
R \geq \log \left[ \frac{\Sigma_x}{2^{-\Delta} \Sigma_x C_2 - \left(\Sigma_x + \Sigma_1  \right) C_1  } \right],\label{eq:vec_region_b_T1_scalar}
\end{align}
with
\begin{align*}
C_1 &=  \frac{\Sigma_0 \Sigma_1 - \Sigma_{01}^2}{(\Sigma_1 - \Sigma_{01})^2},\\
C_2 &= \frac{(\Sigma_x+\Sigma_1)^2}{\Sigma_x \Sigma_1} C_1 + \frac{\Sigma_x}{\Sigma_1}+1.
\end{align*}
These regions in \eqref{eq:vec_region_T1} and \eqref{eq:vec_region_b_T1_scalar} agrees with the relevance-complexity tradeoff of Theorem 3 and Theorem 1 in \cite{tian2009remote} respectively. For the case of a single stage and without side information $\Ym_1=\emptyset$, the region  \eqref{eq:vec_region} reduces to the well known relevance-complexity tradeoff of the Gaussian IB \cite{chechik2005information}. In particular, from \eqref{eq:vec_region_T1} we have
\begin{subequations}\label{eq:vec_region_T1_noSI}
\begin{align}
\Delta  & \leq \log \left|\Id + \Omegam \Sigmam_x \right|\\
\Delta  & \leq R + \log |\Id- \Omegam \Sigmam_{0} |,
\end{align}
\end{subequations}
which reduces to \cite[Theorem 5]{winkelbauer2014ratevector} and from \eqref{eq:vec_region_b_T1_scalar} we obtain
\begin{align}
\label{eq:vec_region_b_T1_scalar_noSI}
R \geq \log\left( \frac{\Sigma_x}{(\Sigma_x+\Sigma_0)2^{-\Delta}-\Sigma_0} \right).
\end{align}
which reduces to \cite[Theorem 2]{winkelbauer2014rate}.
\end{remark}
\begin{figure*}[ht]
     \centering
     \subfloat[ ][Relevance-complexity tradeoff $(R,\Delta_1)$ with $\Delta_2=2$.]{
     \includegraphics[width=0.48\textwidth]{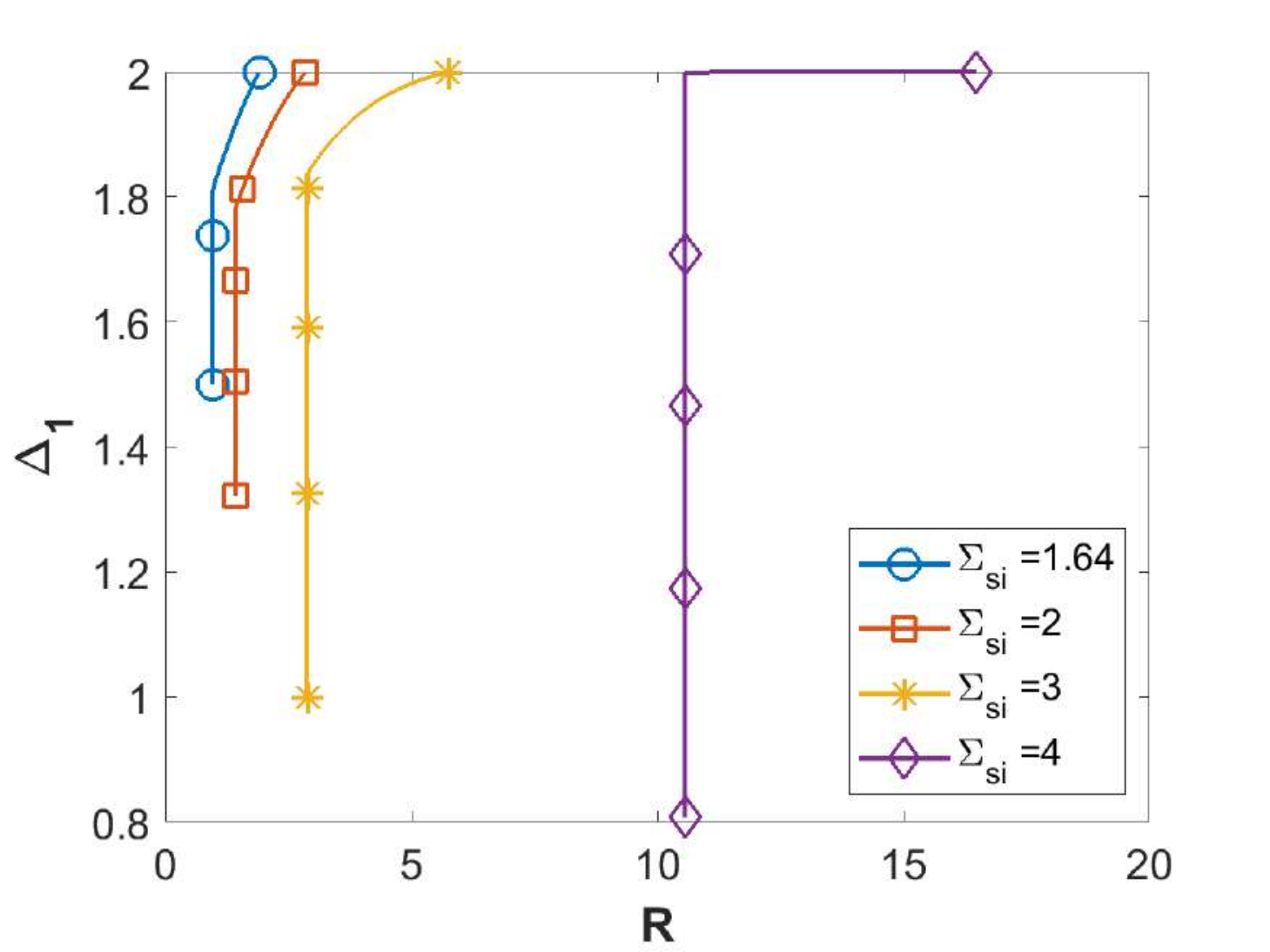}
     \label{fig:GR_SI_01}
     }
     \subfloat[ ][Relevance-complexity tradeoff $(R,\Delta_2)$ with $\Delta_1=1.5$.]{
     \includegraphics[width=0.48\textwidth]{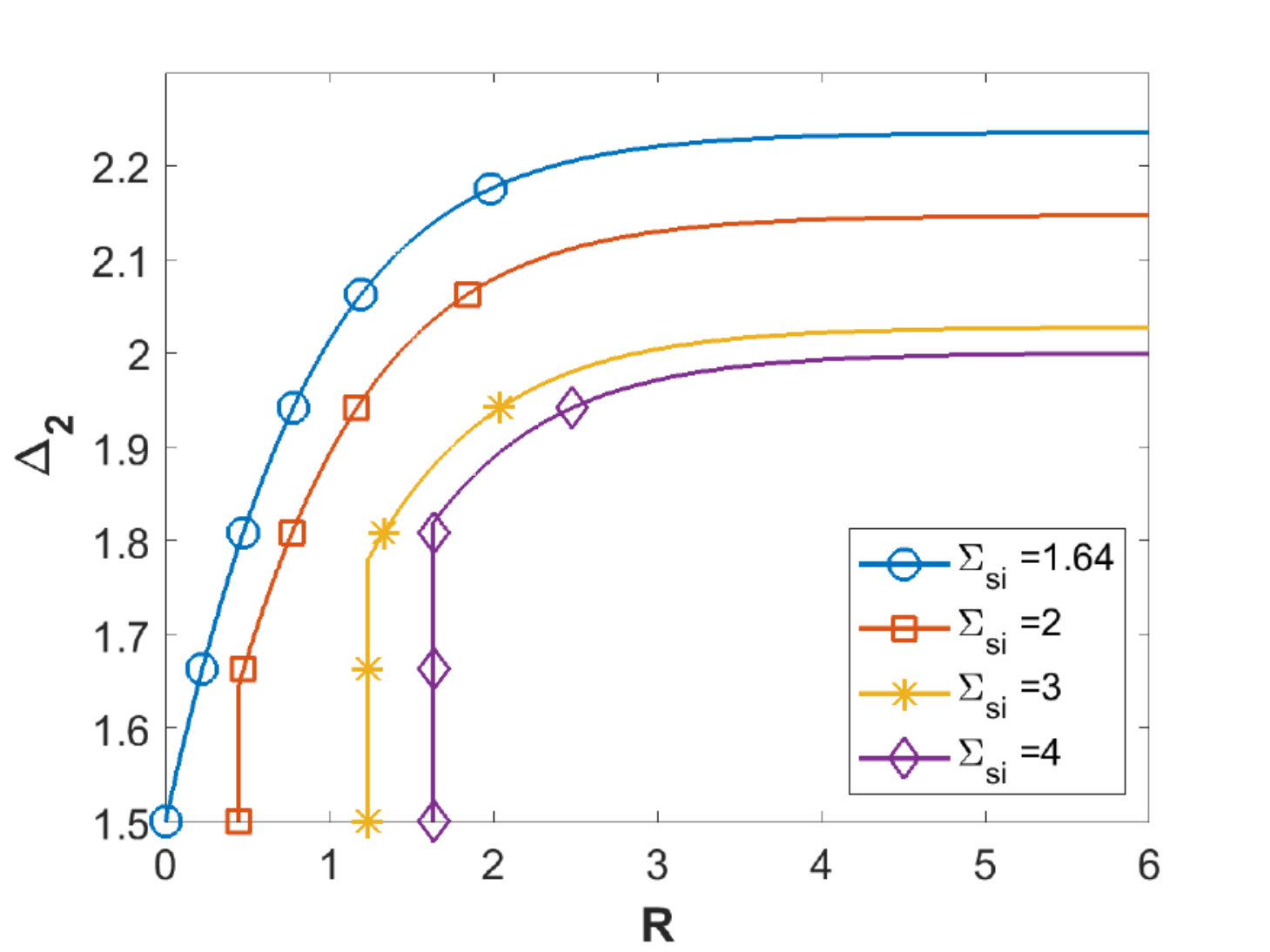}
     \label{fig:GR_SI_02}
     }
     \caption{Comparing the regions $(R,\Delta_1)$ and $(R,\Delta_2)$ for different values of $\Sigma_{\rm si}$.}
     \label{fig:GR_SI}
\end{figure*}

\subsection{Numerical Examples}\label{sect:Scalar Gaussian}
We evaluate the relevance-complexity region \eqref{eq:vec_region} for scalar Gaussian sources and channels with two stages $T=2$ by focusing on the symmetric complexity $R=R_1=R_2$. 
By further letting $\Sigma_{0t} = \gamma \Sigma_t$ for $\gamma\in [0,1]$ and $\Sigma_{0|t} = \Sigma_0 - \gamma^2 \Sigma_t$, the tradeoff between the symmetric complexity $R$ and the relevance $\Delta_1$, $\Delta_2$ for the scalar Gaussian case is given by
\begin{align}\label{eq:SymD}
\Delta_t  \leq \log \left(1 \!+\!  \left[ \Sigma_t^{-1} \!+\! (1 -\gamma)^2  \Omega_t \right] \Sigma_x \right), \forall t=1,2
\end{align}
 and
\begin{align}\label{eq:SymR}
&R \geq \max\left\{ \Delta_1  - \log (1- \Omega_1 \Sigma_{0|1} ) -  \log\left(1 + \frac{\Sigma_x}{\Sigma_1}\right),  \right. \nonumber \\
&\qquad \left. \frac{1}{2} \left( \Delta_2  - \log (1- \Omega_2 \Sigma_{0|2} ) -  \log\left(1 + \frac{\Sigma_x}{\Sigma_2}\right)\right) \right\}.
\end{align}
Focusing on the symmetric side information such that $\Sigma_1= \Sigma_2=\Sigma_{\rm si}$, 
Fig.~\ref{fig:GR_SI} illustrates the tradeoff between $\Delta_1, \Delta_2$ and $R$ for the different values of $\Sigma_{\rm si}$ by letting 
$\Sigma_x=3$, $\Sigma_0=1$, $\gamma=0.25$, and $\Delta_1=1.5$. 
From $0\leq \Omega_1 \leq \Omega_2 \leq \Sigma_{0|2}^{-1}$, it readily follows that we have $1.64\leq \Sigma_{\rm si} \leq 4$. 

In Fig. \ref{fig:GR_SI_01}, we observe the impact of side information noise $\Sigma_{\rm si}$ result to the relevance $\Delta_1$. For each value of $\Sigma_{\rm si}$, we need a minimum value of $R$ such that the total complexity in the second stage $2R$ satisfies our assumption $\Delta_2=2$. For example, for $\Sigma_{\rm si}=2$ (red curve), the minimum value of $R=1.42$ is required in order to achieve the relevance $\Delta_2=2$ in the second stage. On the other hand, the complexity $R=1.42$ results in the maximum relevance $\Delta_1=1.78$ in the first stage. Thus, we can achieve any value of $\Delta_1<1.78$ for the same complexity $R=1.42$. Similarily in Fig. \ref{fig:GR_SI_02}, we observe that there exists a minimum value of $R$ such that the complexity in the first stage $R$ satisfies the relevance at the first stage $\Delta_1=1.5$. For example, with $\Sigma_{\rm si}=4$ (purple), the minimum complexity is $R=1.63$ yielding $\Delta_2=1.82$. Clearly, as the side information noise decreases, the tradeoff improves.

\section{Variational Scalable IB}\label{sect:Variational}
\subsection{Variational Bound and DNN Parameterization}
\newcommand{\tbetav}{\tiny{\betav}}
\newcommand{\tthetav}{\tiny{\thetav}}

So far we assumed that the joint distribution of $(X, Y, \{Y_t\}_t)$ was known. This section addresses a more practical case when the joint distribution is unknown and can be only estimated empirically through a set of training data.  
Since the problem of learning the encoder and the decoder that minimizes the complexity region for $T$ relevance constraints is difficult, we derive a variational bound that enables a neural parameterization of the relevance-complexity region given in \eqref{eq:region2}. For simplicity, we focus on the sum complexity constraint under $T$ relevance constraints by ignoring the side information. 
Our objective is to minimize for a given set of $\betav=(\beta_1, \dots, \beta_T)$
\begin{align}\label{eq:objective}
\Lc_{\tbetav}(\Pm) =I(U_T;Y) + \sum_{t=1}^T \beta_t H(X|U_t)
\end{align}
over a set of pmfs $\Pm\eqdef \{P_{U_1, \dots, U_T | Y}(u_1, \dots, u_T|y) \}$ with $u_t\in \Uc_t$ for any $t$ and $y\in \Yc$. In order to derive a tight variational bound on \eqref{eq:objective}, we consider a set of $T$ arbitrary decoding distributions $\{Q_{X|U_t}(x|u_t)\}_{t=1}^T$ 
for $u_t\in \Uc_t, x\in \Xc$ and
$T$ arbitrary prior distributions $\{Q_{U_t}(u_t)\}_{t=1}^T$. We let $\Qm\eqdef \{Q_{X|U_t}(x|u_t), Q_{U_t}(u_t)\}_{t=1}^T $ denote the set of these pmfs. 
Using similar techniques of \cite{alemi2016deep, kingma2013auto}, it ready follows that the variational upper bound is given by
\begin{equation}\label{eq:vb}
 \Lc^{\rm VB}_{\tbetav}(\Pm, \Qm)\eqdef  D_{KL}(P_{U_T|Y} || Q_{U_T}) - \sum_{t=1}^T \beta_t \EE[ \log Q_{X|U_t}],
\end{equation}
where $D_{KL}( \cdot || \cdot)$ denotes the Kullback-Leiber divergence.  
By adapting \cite[Lemma 1]{estella2021distributed} to our setting, we can prove the following. 
\begin{lemma}\label{lemma-variational}
For fixed pmfs $\Pm$ and $\betav$, we have
\begin{align}
\Lc_{\tbetav}(\Pm)\leq \Lc^{\rm VB}_{\tbetav}(\Pm, \Qm), \;\; \forall \Qm.
\end{align}
Moreover, there exists a unique $\Qm$ that satisfies $\min_{\Qm} \Lc^{\rm VB}_{\tbetav}(\Pm, \Qm) = \Lc_{\tbetav}(\Pm)$, and is given by
\begin{align}\label{lemma2_equality}
Q^*_{X|U_t}=P_{X|U_t}, \;\;\;\;\;\;\;\;\; Q^*_{U_t}=P_{U_t} \;\;\;\;\;\; t\in[T].
\end{align}
\end{lemma}
\begin{proof}
The proof is given in Appendix \ref{proof of lemma2}.
\end{proof}
Now, we present a practical method to minimize \eqref{eq:vb} by parameterizing the encoder $\{P_{U_t|Y}\}_t$ and the decoder $\{Q_{X|U_t}\}_t$ through 
Deep Neural Networks (DNN) parameters $\thetav=(\theta_1, \dots, \theta_T)$ and $\rhov= (\rho_1, \dots, \rho_T)$. 
This enables us to formulate \eqref{eq:vb} in terms of $\thetav, \rhov$ and optimize it using reparameterization trick \cite{kingma2013auto}, Monte Carlo sampling as well as the derivative computation 
$\nabla_{\thetav,\rhov} \Lc^{\rm VB}$.
We let $P_{\thetav}(u_t|y)$ denote the family of the encoding probability distribution $P_{U_t|Y}$ over $\Uc_t$ for each element of $\Yc$, parameterized by the output of a DNN $f_{\theta_t}$ with parameters $\theta_t$ for $t\in [T]$. 
Similarly, let $Q_{\rho_t}(x|u_t)$ denote the family of the $t$-stage decoding distribution $Q_{X|U_t}$ over $\Xc$ for each element of $\Uc_t$, parameterized by the output of a DNN $f_{\rho_t}$ with parameters $\rho_t$ for $t\in [T]$. Finally, we define also
$Q_{\eta_t}(u_t)$ as the family of the prior distributions $Q_{U_t}(u_t)$ over $\Uc_t$ that do not depend on the DNN.  Then, we have  
\begin{align}
\min_{\Pm} \min_{\Qm} \Lc^{\rm VB}_{\tbetav}(\Pm, \Qm) \leq \min_{\thetav, \rhov, \etav} \Lc^{\rm DNN}_{\tbetav}(\thetav, \rhov, \etav)
\end{align}
where we define 
\begin{align}\label{eq:DNN}
 \Lc^{\rm DNN}_{\tbetav}(\thetav, \rhov, \etav)& \eqdef D_{KL}(P_{\thetav}(U_T|Y) || Q_{\eta_T}(U_T)) \nonumber\\
& \;\;\;\; - \sum_{t=1}^T \beta_t \EE[ \log Q_{\rho_t}(X|U_t)].
\end{align}
In order to approximate the objective function \eqref{eq:DNN} through $N$ training data samples $\left\{x_i,y_i \right\}_{i=1}^N$, we generate $M$ independent samples $\{u_{t, i, j}\}_{j=1}^M \sim P_{ \theta_t}(u_t|y)$ such that the 
empirical objective for the $i$-th data sample is given by
  \begin{align}\label{eq:empi}
 \Lc^{\rm emp}_{\tbetav, i}(\thetav, \rhov, \etav)& \eqdef D_{KL}(P_{\thetav}(u_{T, i}|y_i) || Q_{\eta_T}(u_{T,i})) \nonumber\\
&- \frac{1}{M} \sum_{j=1}^M \sum_{t=1}^T \beta_t \EE[ \log Q_{\rho_t}(x_i| u_{t, i, j})].
\end{align}
Finally, we minimize the empirical objective over $N$ training data samples as $\min_{\thetav, \rhov, \etav} \frac{1}{N} \sum_{i=1}^N \Lc^{\rm emp}_{\tbetav, i}(\thetav, \rhov, \etav)$.

\begin{figure*}[ht]
     \centering
     \subfloat[ ][The model trained with half of original data ($\alpha=1/2$).]{
     \includegraphics[width=0.45\textwidth]{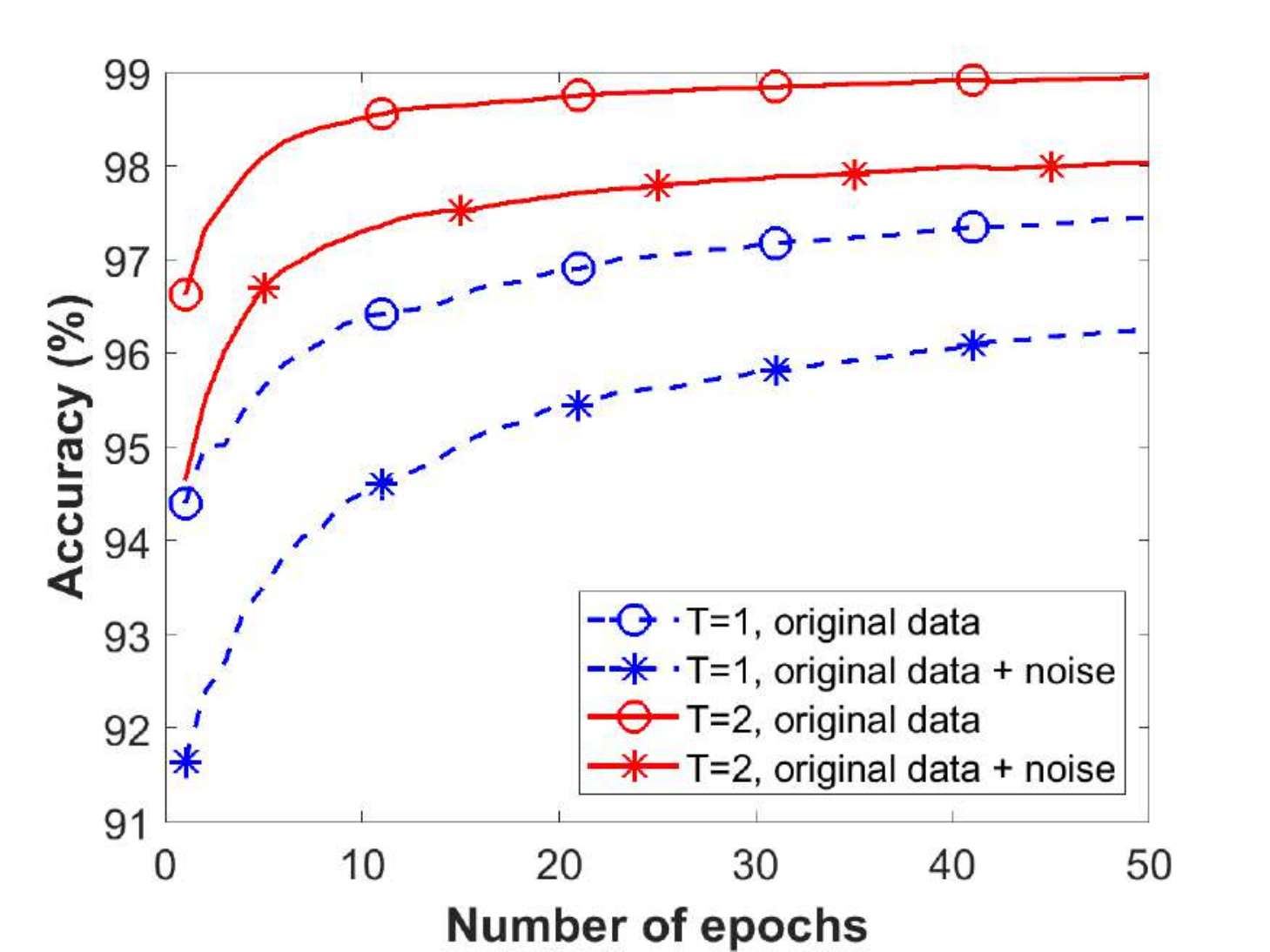}
     \label{fig:GaussSigma3_half}
     }
     \subfloat[ ][The model trained only with original data set ($\alpha=0$).]{
     \includegraphics[width=0.45\textwidth]{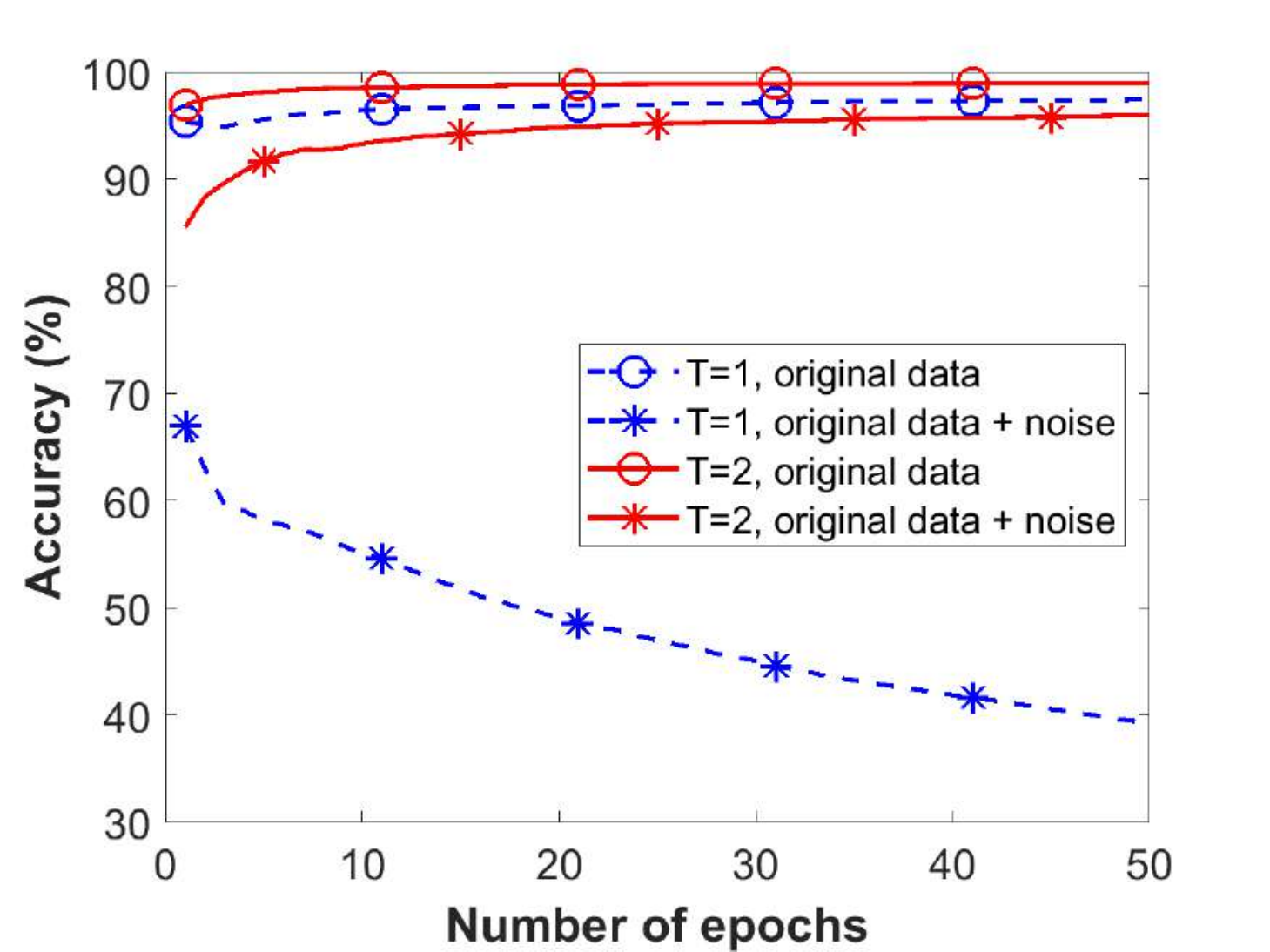}
     \label{fig:GaussSigma3_zero}
     }
     \caption{Accuracy versus the number of epochs.}
     \label{fig:GaussSigma3}
\end{figure*}
\subsection{Experimental Results}
We apply Algorithm \ref{alg:variational} to classify the MNIST dataset \cite{lecun1998gradient}, consisting of 70000 labeled images of handwritten digits between $\{0, \dots, 9\}$. We consider variational scalable information bottleneck with $T=2$ and compare with a baseline with $T=1$. In both scenarios, we train the model with $N=50000$ images such that a fraction $\alpha$ of them are original data contaminated by 
additional noise and the remaining fraction $1-\alpha$ are original data itself, where the noise is modeled as independent Gaussian with zero mean and standard deviation $0.3$ to each pixel and further 
truncated to $[0,1]$ (Fig. \ref{fig:noisy-image}). 
After training the model for each epoch, we test the model with two separate test datasets each containing 10000 images. The first test dataset is noise-free and the second one is noisy with the same noise standard deviation of $0.3$. Moreover, in both scenarios we consider a same standard convolutional neural network (CNN) architecture (Table \ref{table:cnn}) that can achieve the highest accuracy of $99.8\%$ in the noiseless test for $T=1$ and $\alpha=0$. In Table \ref{table:cnn}, we set $\textit{latent dimension} =20$ for $T=1$ and $\textit{latent dimension} =10$ in each stage for $T=2$ for fair comparison.
 \begin{algorithm}
 \caption{Variational Scalable IB Algorithm}
 \label{alg:variational}
 \begin{algorithmic}[1]
 \renewcommand{\algorithmicrequire}{\textbf{Input:}}
 \renewcommand{\algorithmicensure}{\textbf{Output:}}
 \REQUIRE Training dataset $\mathcal{D}$, parameter $\beta_t, t\in[T]$, DNN, parameters $\thetav, \rhov, \etav$.
 \ENSURE  Optimal $\thetav, \rhov, \etav$; tuple $(\sum_{t=1}^T R_t, \Delta_1, \dots, \Delta_T)$.
 \\ \textit{Initialization} : Initialize $\thetav, \rhov, \etav$ and set iteration $k=0$.
  \REPEAT
   	\STATE Randomly select $b$ minibatch samples $X^b$ and the corresponding $Y^b$ from $\mathcal{D}$.
  	\STATE Generate $M$ independent samples $\{u_{t, i, j}\}_{j=1}^M \sim P_{ \theta_t}(u_t|y)$ using reparameterization trick \cite{kingma2013auto}.
  	\STATE Compute $\nabla_{\thetav, \rhov, \etav} \sum_{i=1}^b \Lc^{\rm emp}_{\tbetav, i}(\thetav, \rhov, \etav)$ for $(X^b, Y^b)$.
  	\STATE Update $(\thetav, \rhov, \etav)$ using the estimated gradient.
  \UNTIL {convergence of $(\thetav, \rhov, \etav)$.}
 \end{algorithmic} 
 \end{algorithm}
\begin{figure}
\begin{center}	
\includegraphics[width=0.3\textwidth]{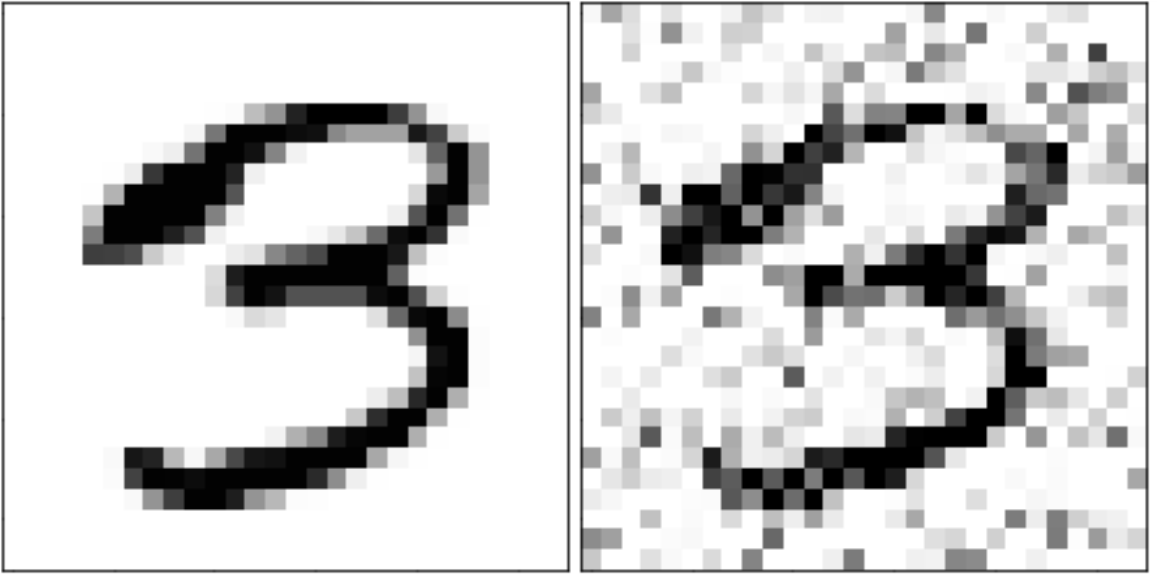}
\caption{An example of MNIST original data (left) and original data with additional noise (right).}
\label{fig:noisy-image}
\end{center}
\end{figure} 	
\begin{table}\caption{DNN architecture}
\vspace*{-0.4em}
\centering
\begin{tabular}{@{}cc@{}}
\toprule
Encoder & DNN Layers                             \\ \midrule
        & conv. ker. {[}5,5,32{]}-Relu           \\
        & maxpool {[}2,2,2{]}                    \\
        & conv. ker. {[}5,5,64{]}-Relu           \\
        & maxpool {[}2,2,2{]}                    \\
        & dense {[}1024{]}-Relu                  \\
        & dropout 0.4                            \\
        & dense {[}$2\times$ latent dimension{]} \\
Decoder & dense {[}100{]}-Relu                   \\
        & dense {[}10{]}-Softmax                 \\ \bottomrule
\end{tabular}
\label{table:cnn}
\vspace*{-2em}
\end{table}

Fig.~\ref{fig:GaussSigma3} compares the classification performance in terms of accuracy in $\%$ versus the number of epochs 
for $T=1,2$. 
In Fig. \ref{fig:GaussSigma3_half}, the model is trained with $\alpha=1/2$. It can be observed that the proposed two-stage scheme ($T=2$) achieves higher accuracy for both of the test datasets. In Fig. \ref{fig:GaussSigma3_zero}, the similar results are obtained when the training is done only with the original data set ($\alpha=0$). In this figure, for $T=1$ and noisy test dataset, the accuracy decreases up to a very low value of $40\%$ because the DNN learns the original train dataset better after each epoch and is unable to predict the noisy test dataset. On the contrary, we observe that the proposed scheme with $T=2$ is able to generalize to unseen noisy data set at the cost of higher complexity. 

\vspace*{-0.2em}
\appendices
\section{Proof of Theorem \ref{Vector_Gaussian_proposition}}\label{proof of vector}
For the achievability, by assuming $\Um_t=\Ym+ \Zm_t$ where $\Zm_t \sim \Nc_{\Cc}(\zerov, \Omegam_{z_t})$  is independent of other variables. We further assume a degraded structure 
$\Omegam_{z_T}\preceq \dots \preceq \Omegam_{z_1}$ satisfying the Markov chain in Remark \ref{Changing Relevance Formulation}. 
For this choice, we look at the relevance term \eqref{eq:region2-relevance}.
\begin{align*}
\Delta & \leq I(\Xm; \Um_t, \Ym_t)\\
&= h(\Um_t, \Ym_t) - h(\Um_t,\Ym_t|\Xm)\\
& =h(\Um_t, \Ym_t) -h(\Wm_0+\Zm_t,\Wm_t)\\
&= \log\frac{|\cov(\Um_t, \Ym_t)|}{|\cov(\Wm_0+\Zm_t,\Wm_t)|} =
 \log|\Id_{2m} +\Am^{-1}\Bm| 
\end{align*}
where we let
\begin{align*}
\Am &= 
\left[
\begin{matrix}
\Sigmam_0 + \Omegam_{z_t} & \Sigmam_{0t}\\
\Sigmam_{t0} & \Sigmam_t 
\end{matrix}\right], \;\;
\Bm= \left[
\begin{matrix} 
\Sigmam_x  & \Sigmam_x \\
\Sigmam_x & \Sigmam_x
\end{matrix}\right]
\end{align*}
By rather straightforward algebra, we can show that $|\Id_{2m} +\Am^{-1}\Bm|$ coincides with the RHS of \eqref{eq:vec_region_a}.

Next, we evaluate the complexity term in \eqref{eq:region1-rate}
\begin{align*}
\sum_{l=1}^{t}R_l  & \geq h( \Um_t|\Ym_t) -h( \Um_t|\Ym, \Ym_t)\\
&\stackrel{(a)}\geq \Delta_t + h( \Um_t, \Ym_t|\Xm)  -h(\Ym_t) -h(\Um_t|\Ym, \Ym_t)\\
&= \Delta_t + h(\Ym_t|\Xm) + h(\Um_t| \Ym_t, \Xm)\\
& \;\;\;\; -h(\Ym_t) -h(\Um_t|\Ym, \Ym_t)\\
&=  \Delta_t +  h(\Wm_t) + h(\Um_t | \Ym_t, \Xm) -h(\Ym_t) -h(\Zm_t )\\
&\stackrel{(b)} = \Delta_t + \log  |(\pi e) \Sigmam_t| + \log  |\pi e (\Sigmam_{0|t}  + \Omegam_{z_t})|\\
&\;\;\;\; - \log  |\pi e (\Sigmam_x+ \Sigmam_t)|-\log |(\pi e) \Omegam_{z_t}|\\
&=\Delta_t 
- \log  | \Sigmam_x \Sigmam_t^{-1}+ \Id| -\log \frac{ 1}{ |\Omegam_{z_t}^{-1}\Sigmam_{0|t}  +\Id|}\\
& \stackrel{(c)}= \Delta_t 
- \log  | \Sigmam_x \Sigmam_t^{-1}+ \Id| -\log |( \Id + \Omegam^{-1}_{z_t} \Sigmam_{0|t} )^{-1}|
\end{align*}
where (a) follows from the relevance terms \eqref{eq:region2-relevance}; (b) 
follows from $h(\Um_t | \Ym_t, \Xm)=h(\Zm_t + \Wm_0| \Wm_t)$; (c) follows from $ \frac{1}{|\Am|} = |\Am^{-1}|$. By choosing $\Omegam_{z_t}= \Omegam_t^{-1} - \Sigmam_{0|t}$ such that  $0 \preceq \Omegam_t \preceq \Sigmam_{0|t}^{-1}$ and $\zerov \preceq \Omegam_{1} \preceq \dots \preceq \Omegam_{T} \preceq \Sigmam_{0|T}^{-1}$, and plugging it in the last expression, we obtain the desired expression \eqref{eq:vec_region_b}, hence 
completes the achievability proof.

In order to prove the converse part, we first provide two useful lemmas.
\begin{lemma}\cite{ugur2020vector,ekrem2014outer,dembo1991information}\label{Fisher-MMSE}
Let $(\bm{X},\bm{Y})$ be a pair of complex random vectors. 
Then we have
\begin{align*}
\log |(\pi e)\bm{J}^{-1}(\bm{X}|\bm{Y})| \leq h(\bm{X}|\bm{Y}) \leq \log |(\pi e)\mmse(\bm{X}|\bm{Y})|
\end{align*}
where $\Jm(\Xm|\Ym) \eqdef \EE[\nabla \log p(\Xm |\Ym)\nabla \log p(\Xm |\Ym)^\H ]$ denotes the conditional Fisher information matrix and $\mmse(\Xm|\Ym)\eqdef \EE[(\Xm-\EE[\Xm|\Ym])(\Xm-\EE[\Xm|\Ym])^\H ]$ denotes the minimum mean square error (MMSE) matrix.
\end{lemma}

\begin{lemma}\cite{ugur2020vector,ekrem2014outer,palomar2005gradient}\label{Fisher-MMSE-connection}
Let $(\Vm_1,\Vm_2)$ be a random vector with finite second moments and $\Zm$ be a Gaussian vector with zero mean and covariance matrix $\Sigmam_z$ which is independent of $(\Vm_1,\Vm_2)$. Then
\begin{align*}
\mmse(\Vm_2|\Vm_1,\Vm_2+\Zm)=\Sigmam_z - \Sigmam_z \Jm (\Vm_2+\Zm|\Vm_1)\Sigmam_z.
\end{align*}

\end{lemma}

By combining the complexity constraints \eqref{eq:region1-rate} and the relevance constraints \eqref{eq:region1-relevance}, we have
\begin{align}
\Delta_t - \sum_{l=1}^{t}R_l &\leq 
I(\Xm;\Ym_t)+I(\Xm;\Um^t|\Ym_t)-I(\Ym;\Um^t|\Ym_t)\nonumber\\
& \stackrel{(a)} = I(\Xm;\Ym_t)-I(\Ym;\Um^t|\Ym_t,\Xm)\nonumber\\
& = I(\Xm;\Ym_t) - h(\Ym|\Xm,\Ym_t) +h(\Ym|\Xm,\Ym_t,\Um^t)\nonumber\\
& \stackrel{(b)} \leq \log|\Id + \Sigmam_x \Sigmam_t^{-1} | - \log|(\pi e) \Sigmam_{0|t}|\nonumber\\
& \;\;\;\;+\log|(\pi e)(\Sigmam_{0|t}-\Sigmam_{0|t}\Omegam_t\Sigmam_{0|t})|\nonumber\\
& =  \log|\Id + \Sigmam_x \Sigmam_t^{-1} | + \log |\Id- \Omegam_t \Sigmam_{0|t} | ,
\end{align}
where (a) follows by the Markov chain $X \mkv (Y, Y_t) \mkv U^t$; 
(b) follows by applying the upper bound of Lemma \ref{Fisher-MMSE} in the third term and noticing the following relation, 
 \[\zerov \preceq \mmse(\Ym|\Xm,\Ym_t,\Um^t) \preceq \mmse(\Ym|\Xm,\Ym_t) = \Sigmam_{0|t},\]
implying that there exists $0 \preceq \Omegam_t \preceq \Sigmam_{0|t}^{-1}$ satisfying $\mmse(\Ym|\Xm,\Ym_t,\Um^t) = \Sigmam_{0|t}-\Sigmam_{0|t}\Omegam_t\Sigmam_{0|t}$. This establishes  \eqref{eq:vec_region_b}.

Now we look at  the relevance constraints \eqref{eq:region1-relevance}. By using Lemma \ref{Fisher-MMSE} we have
\begin{align}
\label{eq:1010}
\Delta_t & \leq I(\Xm;\Um^t, \Ym_t)=h(\Xm)-h(\Xm|\Um^t, \Ym_t)\nonumber\\
&\stackrel{(a)} \leq \log |(\pi e) \Sigmam_x|-\log |(\pi e) \Jm^{-1}(\Xm|\Um^t, \Ym_t)|,
\end{align}
where (a) follows by applying the lower bound in Lemma \ref{Fisher-MMSE}.
%
%
In order to apply Lemma \ref{Fisher-MMSE-connection}, we define first the MMSE estimation $\hat{\Xm}_t = \EE[\Xm|\Ym_{0,t}]$ of $\Xm$ given a $2m$-dimensional vector observation denoted by $\Ym_{0,t}=[ \Ym^\T, \Ym_t^\T]^\T$ such that
\begin{align*}
\Xm =  \hat{\Xm}_t +\tilde{\Xm}_t 
\end{align*}
where $\tilde{\Xm}_t  \sim \Nc_{\Cc}(\zerov, \Sigmam_{\tilde{x}_t})$ denotes the estimation error vector, independent of $\Ym_{0,t}$, with covariance given by 
\begin{align}\label{eq:ErrorCov}
\Sigmam_{\tilde{x}_t}^{-1} = \Sigmam_x^{-1} + \Id_{m \times 2m} \Sigmam_{w_0, w_t}^{-1}\Id_{m \times 2m}^\H 
\end{align}
where $\Id_{m \times 2m} =[\Id_m ~\Id_m]$ and $\Sigmam_{w_0, w_t}$ given by \eqref{eq:covariance matrix}. We define
\begin{align}
\Sigmam_{w_0, w_t}^{-1} = \left[
\begin{matrix} 
\Pm_1  & \Pm_2 \\
\Pm_3 & \Pm_4
\end{matrix}\right],
\end{align}
where $\Pm_1 = \Sigmam_{0|t}^{-1}$, $\Pm_2  = - \Sigmam_{0|t} \Sigmam_{0t} \Sigmam_{t}^{-1}$, $\Pm_3  = - \Sigmam_{t}^{-1} \Sigmam_{t0} \Sigmam_{0|t}^{-1}$,
$\Pm_4  = (\Sigmam_{t} - \Sigmam_{t0}\Sigmam_{0}^{-1}\Sigmam_{0t} )^{-1} = \Sigmam_{t}^{-1}+\Sigmam_{t}^{-1}\Sigmam_{t0}\Sigmam_{0|t}^{-1} \Sigmam_{0t}\Sigmam_{t}^{-1}$.
Using \eqref{eq:ErrorCov}, the estimator can be written as
\begin{align}
\hat{\Xm}_t &= \Sigmam_{\tilde{x}_t}  \Id_{m \times 2m}\Sigmam_{w_0, w_t}^{-1} \Ym_{0,t}\nonumber\\
&=\Sigmam_{\tilde{x}_t}(\Pm_1+\Pm_3)\Ym+\Sigmam_{\tilde{x}_t}(\Pm_2+\Pm_4)\Ym_t,
\label{eq:estimator_converse}
\end{align}
Now, we apply Lemma \ref{Fisher-MMSE-connection} by letting $\Zm=\tilde{\Xm}_t$, $\Vm_1=(\Um^t, \Ym_t)$ and $\Vm_2=\hat{\Xm}_t$ and obtain
\begin{align}
&\Jm(\Xm|\Um^t, \Ym_t)=\Sigmam_{\tilde{x}_t}^{-1}\!-\!\Sigmam_{\tilde{x}_t}^{-1}\mmse(\hat{\Xm}_t|\Um^t, \Ym_t,\Xm) \Sigmam_{\tilde{x}_t}^{-1} \nonumber\\
&\stackrel{(a)} =\Sigmam_{\tilde{x}_t}^{-1}\!-\!\Sigmam_{\tilde{x}_t}^{-1}\mmse(\Sigmam_{\tilde{x}_t}(\Pm_1+\Pm_3)\Ym|\Um^t, \Ym_t,\Xm) \Sigmam_{\tilde{x}_t}^{-1} \nonumber\\
&\stackrel{(b)} = \Sigmam_{\tilde{x}_t}^{-1} - (\Pm_1 \!+\! \Pm_3) \mmse(\Ym |\Um^t, \Ym_t,\Xm)  (\Pm_1 \!+\! \Pm_3)^\H \nonumber\\
& \stackrel{(c)} =\Sigmam_{\tilde{x}_t}^{-1} -  (\Pm_1 \!+\! \Pm_3) (\Sigmam_{0|t} - \Sigmam_{0|t} \Omegam_t \Sigmam_{0|t}) (\Pm_1 \!+\! \Pm_3)^\H,
\label{eq:fisher-converse}
\end{align}
where (a) follows by using \eqref{eq:estimator_converse} and noticing that $\Ym_t=\mathbb{E}(\Ym_t|\Um^t, \Ym_t,\Xm)$; (b) follows by a simple relation $\mmse(\Am \Ym| \cdot) = \Am \mmse(\Ym| \cdot) \Am^\H$ for some $\Am$; (c) follows because $\mmse(\Ym|\Um^t, \Ym_t, \Xm) \preceq \mmse(\Ym|\Ym_t, \Xm) =\Sigmam_{0|t}$ and we can find $\Omegam_t$ satisfying $\bm{0} \preceq \Omegam_{t} \preceq \Sigmam_{0|t}^{-1}$, such that $\mmse(\Ym|\Um^t, \Ym_t, \Xm)=\Sigmam_{0|t} - \Sigmam_{0|t} \Omegam_t \Sigmam_{0|t}$.

Finally, by plugging \eqref{eq:fisher-converse} into \eqref{eq:1010} and with some algebra, we obtain the desired expression \eqref{eq:vec_region_a}, hence completes the converse part.

\section{Proof of Lemma \ref{lemma-variational}}\label{proof of lemma2}
First, we have
\begin{align}
&\EE[-\log Q_{X|U_t}] = D_{KL}(P_{X|U_t}||Q_{X|U_t}) + H(X|U_t).\label{lemma2_01}
\end{align}
Similarily, we have
\begin{align}
I(U_T;Y) &= H(U_T)-H(U_T|Y)\nonumber\\
&= D_{KL}(P_{U_T|Y}||Q_{U_T})-D_{KL}(P_{U_T}||Q_{U_T}). \label{lemma2_02}
\end{align}

Applying \eqref{lemma2_01} to \eqref{eq:vb}, we have
\begin{align}
 \Lc^{\rm VB}_{\tbetav}(\Pm, \Qm) &= D_{KL}(P_{U_T|Y} || Q_{U_T})\nonumber\\
 &+ \sum_{t=1}^T \beta_t [D_{KL}(P_{X|U_t}||Q_{X|U_t}) \!+\! H(X|U_t)]\label{lemma2_03}\\
 & \stackrel{(a)} \geq \Lc_{\tbetav}(\Pm) + \sum_{t=1}^T \beta_t D_{KL}(P_{X|U_t}||Q_{X|U_t})\\
 &  \geq \Lc_{\tbetav}(\Pm),
\end{align}
where (a) follows from $D_{KL}(P_{U_T|Y} || Q_{U_T}) \geq I(U_T;Y)$ and using equation \eqref{eq:objective}. The equality is met by applying \eqref{lemma2_equality} into \eqref{lemma2_02} and \eqref{lemma2_03}.

\bibliographystyle{IEEEbib}
\bibliography{References}

\end{document}